\documentclass[submission,copyright,creativecommons]{eptcs}
\providecommand{\event}{QPL 2016} 
\usepackage{breakurl}             
\title{Operational Meanings of Orders of Observables Defined through 
Quantum Set Theories with Different Conditionals}
\author{Masanao Ozawa\thanks{Supported by JSPS KAKENHI, No.~26247016,
No.~15K13456.}
\institute{Graduate School of Information Science, Nagoya University\\ 
Nagoya, Japan}
\email{ozawa@is.nagoya-u.ac.jp}
}
\def\titlerunning{Quantum Set Theory}
\def\authorrunning{M. Ozawa}
\usepackage{graphicx}
\usepackage{amsthm}
\usepackage{mathptmx}
\usepackage{latexsym}
\usepackage{amsmath,amssymb,graphicx,color,bm}
\usepackage{enumerate}
 \renewcommand{\And}{\wedge}
\newcommand{\cut}[1]{}
\newcommand{\ignore}[1]{}
  \newcommand{\beq}{\begin{equation}}
  \newcommand{\eeq}{\end{equation}}
  \newcommand{\beql}[1]{\begin{equation}\label{eq:#1}}

  \newcommand{\beqa}{\begin{eqnarray}}
  \newcommand{\eeqa}{\end{eqnarray}}
  \newcommand{\beqas}{\begin{eqnarray*}}
  \newcommand{\eeqas}{\end{eqnarray*}}
  \newtheorem{Theorem}{Theorem}[section]
  \newtheorem{Proposition}[Theorem]{Proposition}
  \newtheorem{Lemma}[Theorem]{Lemma}

 \newenvironment{Proof}{\begin{trivlist}
    \item[\hskip \labelsep {\em \indent Proof.}]}{\qed\end{trivlist}}

  \newcommand{\N}{{\bf N}}
  \newcommand{\Q}{{\bf Q}}
  \newcommand{\R}{{\bf R}}
  \newcommand{\V}{{\bf V}}

  \newcommand{\al}{\alpha}
  \newcommand{\be}{\beta}
  \newcommand{\ch}{\chi}

  \newcommand{\et}{\eta}

  \newcommand{\la}{\lambda}
  \newcommand{\mb}{\mbox}
  \newcommand{\nn}{\nonumber}
  \newcommand{\om}{\omega}
  \newcommand{\ph}{\phi}
 \newcommand{\ps}{\psi}
  
  \newcommand{\si}{\sigma}

  \newcommand{\De}{\Delta}
  \newcommand{\Ga}{\Gamma}

  \newcommand{\Eq}[1]{Eq.~(\ref{eq:#1})}
  
  \newcommand{\IFF}{\Leftrightarrow}
  
  \newcommand{\Iff}{\Leftrightarrow}
  
  \newcommand{\Inf}{\bigwedge}

  \newcommand{\Not}{\neg}

  \newcommand{\Or}{\vee}

  \newcommand{\Sup}{\bigvee}
  \newcommand{\THEN}{\Rightarrow}
  \newcommand{\Then}{\Rightarrow}

  \newcommand{\VB}{\V^{(\B)}}

\renewcommand{\V}{V}

  \newcommand{\beqan}{\begin{eqnarray*}}
  \newcommand{\beqar}[1]{\begin{equation}\label{#1}\begin{array}{l}}

  \DeclareMathOperator{\dom}{dom}
  
  \newcommand{\eeqar}{\end{array}\end{equation}}
  \newcommand{\eq}[1]{(\ref{eq:#1})}

\newcommand{\bracket}[1]{\langle#1\rangle}
  \newcommand{\rank}{\mbox{\rm rank}}

\newcommand{\bmat}{\left[\begin{array}{rr}}
\newcommand{\emat}{\end{array}\right]}
\newcommand{\bvec}{\left[\begin{array}{r}}
\newcommand{\evec}{\end{array}\right]}
\newcommand{\btmat}{\left[\begin{array}{rrr}}
\newcommand{\etmat}{\end{array}\right]}

  \newcommand{\cA}{{\cal A}}
  \newcommand{\cB}{{\cal B}}
  \newcommand{\cC}{{\cal C}}
  \newcommand{\cD}{{\dom}}
  
  \newcommand{\cF}{{\cal F}}
  
  \newcommand{\cH}{{\cal H}}

  \newcommand{\cL}{{\cal L}}
  \newcommand{\cM}{{\cal M}}

  \newcommand{\cP}{{\cal P}}
  \newcommand{\cQ}{{\cal Q}}
  \newcommand{\cR}{{\cal R}}

  \newcommand{\tA}{\tilde{A}}

  \newcommand{\tP}{\tilde{P}}
  \newcommand{\tQ}{\tilde{Q}}

  \newcommand{\tX}{\tilde{X}}
  \newcommand{\tY}{\tilde{Y}}

\newcommand{\VQH}{\V^{(\cQ(\cH))}}
\newcommand{\VQ}{\V^{(\cQ)}}
\newcommand{\VL}{\V^{(\cQ)}}
\renewcommand{\VB}{\V^{(\cB)}}
\newcommand{\RQ}{\R^{(\cQ)}}

\newcommand{\LL}{\cL}
\renewcommand{\L}{L}

\renewcommand{\inf}{\bigwedge}
\renewcommand{\sup}{\bigvee}
\renewcommand{\Then}{\rightarrow}
\newcommand{\Thenj}{\rightarrow}
\newcommand{\Thenjj}{\rightarrow_j}
\renewcommand{\Iff}{\leftrightarrow}

\newcommand{\iin}{\in\cR}
\newcommand{\val}[1]{[\![#1]\!]}
\newcommand{\vval}[1]{[\![#1]\!]_{j,\cQ}}
\newcommand{\vvall}[1]{[\![#1]\!]_{j,\cQ(\cH)}}

\newcommand{\valj}[1]{[\![#1]\!]}
 
\newcommand{\tcom}{{}\!\!\perp\!\!\!\!\!\perp\!\!} 
\newcommand{\commutes}{\ {}^{|}\!\!{}_{\circ}\ }
\newcommand{\cuniv}{\underline{\Or}}
\newcommand{\cm}{\underline{\Or}}

\newcommand{\com}{{\tcom}}
\newcommand{\p}{{}^{\perp}}
\newcommand{\hu}{\hat{u}}

\newcommand{\ck}[1]{\check{#1}}
\newcommand{\cx}{\check{x}}
\newcommand{\cy}{\check{y}}

\newcommand{\af}{\,\et\,}
\newcommand{\id}{{\rm id}}

\newcommand{\bitem}{\begin{itemize}\itemsep=0in}
\newcommand{\benum}[1]{\begin{enumerate}[#1]\itemsep=0in}
\newcommand{\eenum}{\end{enumerate}}
\newcommand{\eitem}{\end{itemize}}

\begin{document}
\maketitle
\begin{abstract}
In quantum logic there is well-known arbitrariness in choosing a binary 
operation for conditional. Currently, we have at least three candidates,
called the Sasaki conditional, the contrapositive Sasaki conditional,
and the relevance conditional.  A fundamental problem is to show how the form 
of the conditional follows from an analysis of operational concepts in 
quantum theory.  Here, we attempt such an analysis through quantum set theory (QST).
In this paper, we develop quantum set theory based on quantum logics with those three 
conditionals, each of which defines different quantum logical truth
value assignment.  We show that those three models 
satisfy the transfer principle of the same form to determine the quantum 
logical truth values of theorems of the ZFC set theory.
We also show that the reals in the model and the truth values of their equality 
are the same for those models.
Interestingly, however, the order relation between quantum reals significantly 
depends on the underlying conditionals.  
We characterize the operational meanings of those order relations in terms of 
joint probability obtained by the successive projective measurements 
of arbitrary two observables.
Those characterizations clearly show their individual features 
and will play a fundamental role in future applications to quantum physics.
\end{abstract}

\section{Introduction}
Quantum set theory crosses over two different fields of mathematics, namely, 
foundations of mathematics and foundations of quantum mechanics, 
and originated from the methods of forcing 
introduced by Cohen \cite{Coh63,Coh66}  for 
the independence proof of the continuum hypothesis and quantum logic
introduced by Birkhoff and von Neumann \cite{BvN36}.
After Cohen's work, the forcing subsequently became a central method in 
set theory and also incorporated with various notions in 
mathematics, in particular,  the notion of sheaves \cite{FS79}
and notions of sets in nonstandard logics
such as Boolean-valued set theory \cite{Bel85},
by which Scott and Solovay \cite{SS67} reformulated the method of forcing,
topos \cite{Joh77}, and intuitionistic set theory \cite{Gra79}. 
As a successor of those attempts, quantum set theory, 
a set theory based on the Birkhoff-von Neumann quantum logic,
was introduced by Takeuti \cite{Ta81}, who established the one-to-one correspondence 
between reals in the model (quantum reals)  and quantum observables.
Quantum set theory was recently developed by the present author \cite{07TPQ,16A2}
to obtain the transfer principle to determine quantum truth values of theorems 
of the ZFC set theory,  and clarify the operational meaning 
of the equality between quantum reals, which extends the probabilistic 
interpretation of quantum theory,

In quantum logic there is well-known arbitrariness in choosing a binary 
operation for conditional. Hardegree \cite{Har81} defined a material conditional on 
an orthomodular lattice as a polynomially definable binary operation 
satisfying three fundamental requirements, and showed that there are 
exactly three binary operations satisfying those conditions: the Sasaki 
conditional, the contrapositive Sasaki conditional, and the relevance 
conditional. Naturally, a fundamental problem is to show how the form 
of the conditional follows from an analysis of the operational concept in 
quantum theory testable by experiments.  Here, we attempt such an 
analysis through quantum set theory. In quantum set theory (QST), 
the quantum logical truth values of two atomic formulas, equality and 
membership relations, depend crucially on the choice of conditional. In 
the previous investigations, we have adopted only the Sasaki 
conditional, proved the transfer principle to determine quantum truth 
values of theorems of the ZFC set theory, established the one-to-one 
correspondence between reals in the model, or ``quantum reals'',  and 
quantum observables, and clarified the operational meaning of the 
equality between quantum reals. In this paper, we study QST based on 
the above three material conditional together.
We construct the universal QST model based on the logic of the projection lattice of a 
von Neumann algebra with each conditional.
Then, we show that this new model satisfies the transfer principle of the same form as 
the old model based on the Sasaki conditional. We also show that the 
reals in the model and the truth values of their equality are the same for those three
models. Up to this point, those models behave indistinguishably. 
However, we reveal that the order relation between quantum reals
depend crucially on the underlying conditionals.
We characterize the operational meanings of those order relations,
which turn out closely related to the spectral order introduced by Olson \cite{Ols71}
playing a significant role in the topos approach to quantum theory \cite{DW12},
in terms of joint probability of the outcomes of the successive projective 
measurements of two observables.
Those characterizations clarify their individual features 
and will play a fundamental role in future applications to quantum physics.

\section{Preliminaries}
\subsection{Complete orthomodular lattices}

A {\em complete orthomodular lattice}  is a complete
lattice $\cQ$ with an {\em orthocomplementation},
a unary operation $\perp$ on $\cQ$ satisfying
(C1)  if $P \le Q$ then $Q^{\perp}\le P^{\perp}$,
(C2) $P^{\perp\perp}=P$,
(C3) $P\Or P^{\perp}=1$ and $P\And P^{\perp}=0$,
where $0=\Inf\cQ$ and $1=\Sup\cQ$,
that  satisfies the {\em orthomodular law}
(OM) if $P\le Q$ then $P\Or(P^{\perp}\And Q)=Q$.
In this paper, any complete orthomodular lattice is called a {\em logic}.
A non-empty subset of a logic $\cQ$ is called a {\em subalgebra} iff
it is closed under $\And $, $\Or$, and $\perp$.
A subalgebra $\cA$ of $\cQ$ is said to be {\em complete} iff it has
the supremum and the infimum in $\cQ$ of an arbitrary subset of $\cA$.
For any subset $\cA$ of $\cQ$, 
the subalgebra generated by $\cA$ is denoted by
$\Ga_0\cA$.
We refer the reader to Kalmbach \cite{Kal83} for a standard text on
orthomodular lattices.

We say that $P$ and $Q$ in a logic $\cQ$ {\em commute}, in  symbols
$P\commutes Q$, iff  $P=(P\And Q)\Or(P\And
Q^{\perp})$.  A logic $\cQ$ is a Boolean
algebra if and only if $P\commutes Q$  for all $P,Q\in\cQ$ \cite[pp.~24--25]{Kal83}.
For any subset $\cA\subseteq\cQ$,
we denote by $\cA^{!}$ the {\em commutant} 
of $\cA$ in $\cQ$ \cite[p.~23]{Kal83}, i.e., 
\[
\cA^{!}=
\{P\in\cQ\mid P\commutes Q \mbox{ for all }
Q\in\cA\}.
\]
Then, $\cA^{!}$ is a complete subalgebra of $\cQ$.
A {\em sublogic} of $\cQ$ is a subset $\cA$ of
$\cQ$ satisfying $\cA=\cA^{!!}$. 
For any subset $\cA\subseteq\cQ$, the smallest 
logic including $\cA$ is 
$\cA^{!!}$ called the  {\em sublogic generated by
$\cA$}.
Then, it is easy to see that a subset 
 $\cA$ is a Boolean sublogic, or equivalently 
 a distributive sublogic, if and only if 
$\cA=\cA^{!!}\subseteq\cA^{!}$.

The following proposition is useful in later discussions.

\begin{Proposition}\label{th:logic}
Let $\cQ$ be a  logic on $\cH$.
If $P_{\al}\in\cQ$ and 
$P_{\al}\commutes Q$ for all $\al$, then
$(\Sup_{\al}P_{\al})\commutes Q$, 
$\Inf_{\al}P_{\al}\commutes Q$,
$Q \And (\Sup_{\al}P_{\al})=\Sup_{\al}(Q\And
P_{\al})$C
$Q \And (\Inf_{\al}P_{\al})=\Inf_{\al}(Q\And
P_{\al})$.
\end{Proposition}
\begin{Proof}
Suppose that $P_{\al}\in\cQ$ and $P_{\al}\commutes Q$ hold for every $\al$.
From
\[
\Sup_{\al}P_\al\And  Q\le Q,\quad
\Sup_{\al}P_\al\And  Q^\perp\le Q^{\perp},
\]
we have
\beqa\label{eq:sup-com}
\Sup_{\al}P_\al\And  Q\commutes Q,\quad
\Sup_{\al}P_\al\And  Q^\perp\commutes Q.
\eeqa
By the assumption, we have
$P_{\al}=(P_\al\And  Q)\Or(P_\al\And  Q^\perp)$ for every $\al$.
Since 
\beqas
\Sup_{\al}P_{\al}&=&
\Sup_{\al}(P_\al\And  Q)\Or(P_\al\And  Q^\perp)\\
&=&
(\Sup_{\al}P_\al\And  Q)\Or(\Sup_{\al}P_\al\And  Q^\perp),
\eeqas
by \Eq{sup-com} we have $\Sup_{\al}P_{\al}\commutes Q$.
By \Eq{sup-com} the distributive law holds and we have
\beqas
Q\And\Sup_{\al}P_{\al}
&=&
Q\And[(\Sup_{\al}P_\al\And Q)\Or(\Sup_{\al}P_\al\And  Q^\perp)]\\
&=&
\Sup_{\al}(P_\al\And  Q).
\eeqas
Thus, we have
$Q\And\Sup_{\al}P_{\al}=
\Sup_{\al}(Q\And P_\al)$.
The rest of the assertions follows from 
the De Morgan law.
\end{Proof}

\subsection{Logics on Hilbert spaces}
Let $\cH$ be a Hilbert space.
For any subset $S\subseteq\cH$,
we denote by $S^{\perp}$ the orthogonal complement
of $S$.
Then, $S^{\perp\perp}$ is the closed linear span of $S$.
Let $\cC(\cH)$ be the set of all closed linear subspaces in
$\cH$. 
With the set inclusion ordering, 
the set $\cC(\cH)$ is a complete
lattice. 
The operation $M\mapsto M^\perp$ 
is  an orthocomplementation
on the lattice $\cC(\cH)$, with which $\cC(\cH)$ is a logic.

Denote by $\cB(\cH)$ the algebra of bounded linear
operators on $\cH$ and $\cQ(\cH)$ the set of projections on $\cH$.
We define the {\em operator ordering} on $\cB(\cH)$ by
$A\le B$ iff $(\ps,A\ps)\le (\ps,B\ps)$ for
all $\ps\in\cH$. 
For any $A\in\cB(\cH)$, denote by $\cR (A)\in\cC(\cH)$
the closure of the range of $A$, {\em i.e.,} 
$\cR(A)=(A\cH)^{\perp\perp}$.
For any $M\in\cC(\cH)$,
denote by $\cP (M)\in\cQ(\cH)$ the projection operator 
of $\cH$
onto $M$.
Then, $\cR\cP (M)=M$ for all $M\in\cC(\cH)$
and $\cP\cR (P)=P$ for all $P\in\cQ(\cH)$,
and we have $P\le Q$ if and only if $\cR (P)\subseteq\cR (Q)$
for all $P,Q\in\cQ(\cH)$,
so that $\cQ(\cH)$ with the operator ordering is also a logic
isomorphic to $\cC(\cH)$.
Any sublogic of $\cQ(\cH)$ will be called a {\em logic on $\cH$}. 
The lattice operations are characterized by 
$P\And Q={\mb{weak-lim}}_{n\to\infty}(PQ)^{n}$, 
$P^\perp=1-P$ for all $P,Q\in\cQ(\cH)$.

Let $\cA\subseteq\cB(\cH)$.
We denote by $\cA'$ the {\em commutant of 
$\cA$ in $\cB(\cH)$}.
A self-adjoint subalgebra $\cM$ of $\cB(\cH)$ is called a
{\em von Neumann algebra} on $\cH$ iff 
$\cM''=\cM$.
For any self-adjoint subset $\cA\subseteq\cB(\cH)$,
$\cA''$ is the von Neumann algebra generated by $\cA$.
We denote by $\cP(\cM)$ the set of projections in
a von Neumann algebra $\cM$.
For any $P,Q\in\cQ(\cH)$, we have 
$P\commutes Q$ iff $[P,Q]=0$, where $[P,Q]=PQ-QP$.
For any subset $\cA\subseteq\cQ(\cH)$,
we denote by $\cA^{!}$ the {\em commutant} 
of $\cA$ in $\cQ(\cH)$.
For any subset $\cA\subseteq\cQ(\cH)$, the smallest 
logic including $\cA$ is the logic
$\cA^{!!}$ called the  {\em logic generated by
$\cA$}.  
Then, a subset $\cQ \subseteq\cQ(\cH)$ is a logic on $\cH$ if
and only if $\cQ=\cP(\cM)$ for some von Neumann algebra
$\cM$ on $\cH$ \cite[Proposition 2.1]{07TPQ}.

\subsection{Commutators}

Marsden \cite{Mar70} has introduced the commutator $\com(P,Q)$ 
of two elements $P$ and $Q$ of a logic $\cQ$ by 
\[
\com(P,Q)=(P\And Q)\Or(P\And Q\p)\Or(P\p\And Q)\Or(P\p\And Q\p).
\]
Bruns and Kalmbach \cite{BK73} have generalized this notion 
to finite subsets of $\cQ$ by 
\[
\com(\cF)=\Sup_{\al:\cF\to\{\id,\perp\}}\Inf_{P\in\cF}P^{\al(P)}
\]
for all $\cF\in\cP_{\om}(\cQ)$,
where $\cP_{\om}(\cQ)$ stands for the set of finite subsets of $\cQ$, and
$\{\id,\perp\}$ stands for the set consisting of the identity operation $\id$ and the 
orthocomplementation~$\perp$.
Generalizing this notion to arbitrary subsets $\cA$ of $\cQ$, Takeuti \cite{Ta81} defined
$\com(\cA)$ by
\[
\com(\cA)=\Sup \{E\in\cA^{!} \mid P_{1}\And E\commutes P_{2}\And E
\mb{ for all }P_{1},P_{2}\in\cA\},
\]
of any $\cA\in\cP(\cQ)$, where $\cP(\cQ)$ stands for the power set of $\cQ$.
Takeuti's definition has been reformulated in several more convenient forms  
\cite{Pul85,Che89,16A2}.

We have the following characterizations of commutators 
 in logics on Hilbert spaces \cite[Theorems 2.5, 2.6, Proposition 2.2]{07TPQ}.

\begin{Theorem}\label{th:com}
Let $\cQ$ be a logic on $\cH$
and let $\cA\subseteq\cQ$.
Then, we have the following relations.
\begin{enumerate}[\rm (i)]\itemsep=0in
\item
$\com(\cA)=\cP\{\ps\in\cH\mid 
[P_1,P_2]P_3\psi=0\mb{ for all } P_1,P_2,P_3\in \cA\}$.
\item
$\com(\cA)=\cP\{\ps\in\cH\mid [A,B]\ps=0 \mb{ for all }A,B\in\cA''\}$.
\eenum
\end{Theorem}

\section{Conditionals}
\label{se:GIIQL}

In classical logic, the conditional operation 
$\Then$ is defined by negation $\perp$ and
disjunction $\Or$ as $P\Then Q=P^{\perp}\Or Q$.
In quantum logic there is a well-known arbitrariness in choosing 
a binary operation for conditional.
Hardegree \cite{Har81} defined a {\em material conditional} 
on an orthomodular lattice $\cQ$ as a polynomially definable binary operation
$\Then$ on $\cQ$ satisfying the following ``minimum implicative conditions'':
\begin{enumerate}[(i)]\itemsep=0in
\item[(LB)] If $P\commutes Q$, then 
$P\Then Q=P^{\perp}\Or Q$ for all $P,Q\in\cQ$.
\item[(E)]  $P\Then Q=1$ if and only if $P\le Q$. 
\item[(MP)] ({\it modus ponens}) $P\And(P\Then Q)\le Q$.
\item[(MT)] ({\it modus tollens}) $Q^{\perp}\And (P\Then Q) \le P^{\perp}$.
\eenum
Then, he proved that there are exactly three material conditionals:
\begin{enumerate}[(i)]\itemsep=0in
\item[(S)] (Sasaki conditional) $P\Then{}_{S}Q: =P^{\perp}\Or(P\And Q)$, 
\item[(C)] (Contrapositive Sasaki conditional) $P\Then{}_{C}Q: =(P\Or Q)^{\perp}\Or Q$,
\item[(R)] (Relevance conditional) $P\Then{}_{R}Q: =(P\And Q)\Or(P^{\perp}\And Q)\Or(P^{\perp}\And Q^{\perp})$.
\eenum

We shall denote by $\Thenjj$ with $j=S,C,R$ any one of the above material conditionals.
Once the conditional $\Then_{j}$ is specified, the logical equivalence $\Iff_{j}$ is defined
by
\[
P\Iff_{j} Q := (P\Then_{j} Q)\And(Q\Then_{j} P).
\]
Then, it is easy to see that we have
\[
P\Iff_{S} Q = P\Iff_{C} Q= P\Iff_{R} Q=(P\And Q)\Or (P^{\perp}\And Q^{\perp}).
\]
Thus, we write $\Iff$ for $\Iff_{j}$ for all $j=S,C,R$.

In the previous investigations \cite{Ta81,07TPQ,16A2} on quantum set theory 
only the Sasaki arrow was adopted as the conditional. 
In this paper, we develop a quantum set theory based on the above three 
conditionals together and show that they equally 
ensure the transfer principle for quantum set theory.
We shall also show that the notions of equality defined through those three 
are the same,  but that the order relations defined through them are different.

\newcommand{\ThenR}{\rightarrow_{R}}
\newcommand{\ThenS}{\rightarrow_{S}}

We have the following characterizations of conditionals in logics on Hilbert spaces.

\begin{Theorem}\label{th:com}
Let $\cQ$ be a logic on $\cH$
and let $P,Q\in\cQ$.
Then, we have the following relations.
\begin{enumerate}[\rm (i)]\itemsep=0in
\item $P{\Then}_S Q=\cP\{\ps\in\cH\mid Q^\perp P\ps=0\}$.
\item $P{\Then}_C Q=\cP\{\ps\in\cH\mid PQ^\perp\ps=0\}$.
\item $P{\Then}_R Q=\cP\{\ps\in\cH\mid Q^\perp P\ps=PQ^\perp\ps=0\}$.
\item $P{\Iff} Q=\cP\{\ps\in\cH\mid P\ps=Q\ps\}$.
\eenum
\end{Theorem}
\begin{Proof}
To show (i) suppose $\ps\iin (P\ThenS Q)$.  Then, we have
$\ps=P^{\perp}\ps+(P\And Q)\ps$,
so that we have $Q^\perp P\ps=0$.
Conversely, suppose $Q^\perp P\ps=0$.
Then, we have $P\ps=QP\ps+Q^{\perp}P\ps=QP\ps\in\cR(Q).$
Since $P\ps\in\cR(P)$, we have $P\ps\in\cR(P)\cap\cR(Q)=\cR(P\And Q)$.
It follows that $\ps=P^{\perp}\ps+P\ps=P^{\perp}\ps+(P\And Q)\ps\in\cR(P{\Then}_S Q)$. Thus, relation (i) holds.
Relation (ii) follows from the relation $P{\Then}_C Q=Q^{\perp}\ThenS P^{\perp}$.
Relation (iii) follows from the relation 
$P{\Then}_R Q=(P{\Then}_S Q)\And (P{\Then}_C Q)$.
To show relation (iv), suppose $\ps\in \cR(P{\Iff} Q)$.
Then, $\ps\in\cR(P{\Then}_R Q)\cap\cR(Q{\Then}_R P)$, and hence
$ PQ^\perp\ps=0$ and $P^\perp Q\ps=0$, so that $P\ps=PQ\ps=Q\ps$.  
Conversely, if $P\ps=Q\ps$, we have $Q^\perp P\ps=0$ and $P^\perp Q\ps=0$,
so that $\ps\in\cR(P{\Iff} Q)$.  Thus, relation (iv) follows.
\end{Proof}

The following theorem shows important properties of material conditionals
in establishing the transfer principle for quantum set theory.

\begin{Proposition}\label{th:GC}
The material conditionals $\Thenjj$ with $j=S,C,R$ satisfy the following properties.
\begin{enumerate}[\rm (i)]\itemsep=0in
\item $P\Thenjj Q\in\{P,Q\}^{!!}$ for all $P,Q\in\cQ$.
\item $(P\Thenjj Q)\And E
=[(P\And E)\Then(Q\And E)]\And E$ if $P, Q\commutes E$ for all $P,Q,E\in\cQ$.
\end{enumerate}
\end{Proposition}
\begin{Proof}
Assertions follow from the Lemma below.
\end{Proof}

\begin{Lemma}
\label{th:restriction_property}
Let $f$ be a two-variable ortholattice polynomial on a logic $\cQ$ on $\cH$.
Then, we have the following statements.
\begin{enumerate}
\item[\rm (i)] $f(P,Q)\in\{P,Q\}^{!!}$ for all $P,Q\in\cQ$.
\item[\rm (ii)] $
f(P,Q)\And E=f(P\And E,Q\And E)\And E$ if $P, Q\commutes E$ for all $P,Q,E\in\cQ$.
\end{enumerate}
\end{Lemma}
\begin{Proof}
Since $f(P,Q)$ is in the ortholattice $\Ga_0\{P,G\}$ generated by $P$ and $Q$
and we have $\Ga_0\{P,G\}\subseteq\{P,Q\}^{!!}$, so that statement (i) follows.
The proof of (ii) is carried out by induction on the complexity of the polynomial $f(P,Q)$.
First, note that from $P,Q\commutes E$ we have $g(P,Q)\commutes E$
for any two-variable polynomial $g$.  If $f(P,Q)=P$ or $f(P,Q)=Q$, 
assertion (ii) holds obviously.  If $f(P,Q)=g_1(P,Q)\And g_2(P,Q)$
with two-variable polynomials $g_1,g_2$, the assertion holds from associativity.
Suppose that  $f(P,Q)=g_1(P,Q)\Or g_2(P,Q)$
with two-variable polynomials $g_1,g_2$.
Since $g_1(P,Q), g_2(P,Q)\commutes E$, the assertion follows from the 
distributive law focusing on $E$.
Suppose $f(P,Q)=g(P,Q)^\perp$ with a two-variable polynomial $g$.  
For the case where $g$ is atomic, the assertion follows; for instance, if $g(P,Q)=P$, 
we have 
$f(P\And E,Q\And E)\And E=
(P\And E)^\perp\And E=
(P^\perp\Or E^\perp)\And E=P^\perp\And E
=f(P,Q)\And E$.
Then, we assume $g(P,Q)=g_1(P,Q)\And g_2(P,Q)$
or $g(P,Q)=g_1(P,Q)\Or g_2(P,Q)$ with two-variable 
polynomials $g_1,g_2$.  If
$g(P,Q)=g_1(P,Q)\And g_2(P,Q)$, 
by the induction hypothesis and the distributivity
we have
\beqas
f(P,Q)\And E
&=&
g(P,Q)^\perp\And E\\
&=&
(g_1(P,Q)^\perp\Or g_2(P,Q)^\perp)\And E\\
&=&
(g_1(P,Q)^\perp\And E)
\Or 
(g_2(P,Q)^\perp\And E)\\
&=&
(g_1(P\And E,Q\And E)^\perp\And E)
\Or 
(g_2(P\And E,Q\And E)^\perp\And E)\\
&=&
(g_1(P\And E,Q\And E)^\perp
\Or 
g_2(P\And E,Q\And E)^\perp)\And E)\\
&=&
(g_1(P\And E,Q\And E)
\And
g_2(P\And E,Q\And E))^\perp\And E\\
&=&
g(P\And E,Q\And E)^\perp\And E\\
&=&
f(P\And E,Q\And E)\And E.
\eeqas
Thus, the assertion follows if 
$g(P,Q)=g_1(P,Q)\And g_2(P,Q)$,
and similarly the assertion follows if 
$g(P,Q)=g_1(P,Q)\Or g_2(P,Q)$.
Thus, the assertion generally follows from the induction on the complexity of
the polynomial $f$.
\end{Proof}

\renewcommand{\Then}{\rightarrow_{j}}

\section{Quantum set theory}
\label{se:UQ}

We denote by $\V$ the universe 
of the Zermelo-Fraenkel set theory
with the axiom of choice (ZFC).
Let $\cL(\in)$ be the language 
for first-order theory with equality 
augmented by a binary relation symbol
$\in$, bounded quantifier symbols $\forall x\in y$,
$\exists x \in y$, and no constant symbols.
For any class $U$, 
the language $\cL(\in,U)$ is the one
obtained by adding a name for each element of $U$.

Let $\cQ$ be a logic on $\cH$.
For each ordinal $ {\al}$, let
\[
\V_{\al}^{(\cQ)} = \{u|\ u:\dom(u)\to \cQ \mbox{ and }
(\exists \be<\al)
\dom(u) \subseteq V_{\be}^{(\cQ)}\}.
\]
The {\em $\cQ$-valued universe} $\VL$ is defined
by 
\[
  \VL= \bigcup _{{\al}{\in}\mbox{On}} V_{{\al}}^{(\cQ)},
\]
where $\mbox{On}$ is the class of all ordinals. 
For every $u\in\VQ$, the rank of $u$, denoted by
$\rank(u)$,  is defined as the least $\al$ such that
$u\in \VQ_{\al+1}$.
It is easy to see that if $u\in\dom(v)$ then 
$\rank(u)<\rank(v)$.

In what follows $\Then$ generally denotes one of 
the Sasaki conditional $\Thenj_S$, the contrapositive Sasaki 
conditional $\Thenj_C$, and the relevance conditional $\Thenj_R$.
For any $u,v\in\VL$, the $\cQ$-valued truth values of
atomic formulas $u=v$ and $u\in v$ are assigned
by the following rules recursive in rank.
\begin{enumerate}[(i)]\itemsep=0in
\item $\vval{u = v}
= \inf_{u' \in  \cD(u)}(u(u') \Then
\vval{u'  \in v})
\And \inf_{v' \in   \cD(v)}(v(v') 
\Then \vval{v'  \in u})$.
\item $ \vval{u \in v} 
= \sup_{v' \in \cD(v)} (v(v')\And \vval{u =v'})$.
\end{enumerate}

To each statement $\ph$ of $\cL(\in,\VL)$ 
we assign the
$\cQ$-valued truth value $ \val{\ph}_{j,\cQ}$ by the following
rules.
\begin{enumerate}[(i)]\itemsep=0in
\setcounter{enumi}{2}
\item $ \vval{\Not\ph} = \vval{\ph}^{\perp}$.
\item $ \vval{\ph_1\And\ph_2} 
= \vval{\ph_{1}} \And \vval{\ph_{2}}$.
\item $ \vval{\ph_1\Or\ph_2} 
= \vval{\ph_{1}} \Or \vval{\ph_{2}}$.
\item $ \vval{\ph_1\rightarrow\ph_2} 
= \vval{\ph_{1}} \Then \vval{\ph_{2}}$.
\item $ \vval{\ph_1\Iff\ph_2} 
= \vval{\ph_{1}} \Iff \vval{\ph_{2}}$.
\item $ \vval{(\forall x\in u)\, {\ph}(x)} 
= \Inf_{u'\in \dom(u)}
(u(u') \Then \vval{\ph(u')})$.
\item $ \vval{(\exists x\in u)\, {\ph}(x)} 
= \Sup_{u'\in \dom(u)}
(u(u') \And \vval{\ph(u')})$.
\item $ \vval{(\forall x)\, {\ph}(x)} 
= \Inf_{u\in \VL}\vval{\ph(u)}$.
\item $ \vval{(\exists x)\, {\ph}(x)} 
= \Sup_{u\in \VL}\vval{\ph(u)}$.
\end{enumerate}

A formula in $\cL(\in)$ is called a {\em
$\De_{0}$-formula}  if it has no unbounded quantifiers
$\forall x$ or $\exists x$.
The following theorem holds.

\sloppy
\begin{Theorem}[$\De_{0}$-Absoluteness Principle]
\label{th:Absoluteness}
\sloppy  
For any $\De_{0}$-formula 
${\ph} (x_{1},{\ldots}, x_{n}) $ 
of $\cL(\in)$ and $u_{1},{\ldots}, u_{n}\in \VQ$, 
we have
\[
\vval{\ph(u_{1},\ldots,u_{n})}=
\val{\ph(u_{1},\ldots,u_{n})}_{j,\cQ(\cH)}.
\]
\end{Theorem}
\begin{Proof}
The assertion is proved by the induction on the complexity
of formulas and the rank of elements of $\VQ$.
Let $u,v\in\VL$.
\cut{
From Theorem \ref{th:cuniv} and Proposition \ref{th:sublogic}, 
we have
$\vval{\cuniv(u,v)}=\val{\cuniv(u,v)}_{\cQ(\cH)}$.
}
We assume that the assertion holds for all $u'\in\dom(u)$
and $v'\in\dom(v)$.
Then, we have $\vval{u'\in v}=\vvall{u'\in v}$,
$\vval{v'\in u}=\vvall{v'\in u}$,
and $\vval{u=v'}=\vvall{u=v'}$.
Thus, 
\beqas
\vval{u=v}
&=&\Inf_{u'\in\dom(u)}(u(u')\Then\vval{u'\in v})
\And
\Inf_{v'\in\dom(v)}(v(v')\Then\vval{v'\in u})\\
&=&\Inf_{u'\in\dom(u)}(u(u')\Then\vvall{u'\in v})
\And
\Inf_{v'\in\dom(v)}(v(v')\Then\vvall{v'\in u})\\
&=&
\vvall{u=v},
\eeqas
and  we also have
\beqas
\vval{u\in v}
&=&
\Sup_{v'\in\dom(v)}(v(v')\And\vval{u=v'})\\
&=&
\Sup_{v'\in\dom(v)}(v(v')\And\vvall{u=v'})\\
&=&
\vvall{u\in v}.
\eeqas
Thus, the assertion holds for atomic formulas.
Any induction step adding a logical symbol works
easily, even when bounded quantifiers are concerned,
since the ranges of the supremum and the infimum 
are common for evaluating $\vval{\cdots}$ and 
$\vvall{\cdots}$. 
\end{Proof}

Henceforth, 
for any $\De_{0}$-formula 
${\ph} (x_{1},{\ldots}, x_{n}) $
and $u_1,\ldots,u_n\in\VQ$,
we abbreviate $\val{\ph(u_{1},\ldots,u_{n})}_j=
\vval{\ph(u_{1},\ldots,u_{n})}$,
which is the common $\cQ(\cH)$-valued truth value for 
$u_{1},\ldots,u_{n}\in\VQ$.

\renewcommand{\val}[1]{[\![#1]\!]_j}

The universe $\V$  can be embedded in
$\VQ$ by the following operation 
$\vee:v\mapsto\check{v}$ 
defined by the $\in$-recursion: 
for each $v\in\V$, $\check{v} = \{\check{u}|\ u\in v\} 
\times \{1\}$. 
Then we have the following.
\begin{Theorem}[$\De_0$-Elementary Equivalence Principle]
\label{th:2.3.2}
\sloppy  
For any $\De_{0}$-formula 
${\ph} (x_{1},{\ldots}, x_{n}) $ 
of $\cL(\in)$ and $u_{1},{\ldots}, u_{n}\in V$,
we have
$
\bracket{\V,\in}\models {\ph}(u_{1},{\ldots},u_{n})
\mbox{ if and only if }
\val{\ph(\check{u}_{1},\ldots,\check{u}_{n})}=1.
$
\end{Theorem}
\begin{Proof}
Let ${\bf 2}$ be the sublogic such that ${\bf 2}=\{0,1\}$.
Then, by induction it is easy to see that 
$
\bracket{\V,\in}\models  {\ph}(u_{1},{\ldots},u_{n})
\mbox{ if and only if }
\valj{\ph(\check{u}_{1},\ldots,\check{u}_{n})}_{j,\bf 2}=1
$
for any ${\ph} (x_{1},{\ldots}, x_{n})$ in $\LL(\in)$, 
and this is
equivalent  to $\val{\ph(\check{u}_{1},\ldots,\check{u}_{n})}=1$
for any $\De_{0}$-formula ${\ph} (x_{1},{\ldots}, x_{n})$ 
by the $\De_0$-absoluteness principle.
\end{Proof}

\section{Transfer principle}
\label{se:ZFC}\label{se:TPQ}

In this section, we investigate the transfer principle that transfers any $\De_0$-formula
provable in ZFC to a true statement about elements of $\VQ$.

The results in this section have been obtained for $j=S$ in Ref.~\cite{07TPQ}.
Here, we generalize them to the case $j=C,R$.
For $u\in\VQ$, we define the {\em support} 
of $u$, denoted by $L(u)$, by transfinite recursion on the 
rank of $u$ by the relation
\[
L(u)=\bigcup_{x\in\dom(u)}L(x)\cup\{u(x)\mid x\in\dom(u)\}.
\]
For $\cA\subseteq\VQ$ we write 
$L(\cA)=\bigcup_{u\in\cA}L(u)$ and
for $u_1,\ldots,u_n\in\VQ$ we write 
$L(u_1,\ldots,u_n)=L(\{u_1,\ldots,u_n\})$.

For $u\in\VQ$, we define the {\em support} 
of $u$, denoted by $L(u)$, by transfinite recursion on the 
rank of $u$ by the relation
\[
L(u)=\bigcup_{x\in\dom(u)}L(x)\cup\{u(x)\mid x\in\dom(u)\}.
\]
For $\cA\subseteq\VQ$ we write 
$L(\cA)=\bigcup_{u\in\cA}L(u)$ and
for $u_1,\ldots,u_n\in\VQ$ we write 
$L(u_1,\ldots,u_n)=L(\{u_1,\ldots,u_n\})$.
Then, we obtain the following characterization of
subuniverses of $V^{(\cQ(\cH))}$.

\begin{Proposition}\label{th:sublogic}
Let $\cQ$ be a logic on $\cH$ and $\al$ an
ordinal. For any $u\in V^{(\cQ(\cH))}$, we have
$u\in\VL_{\al}$  if and only if
$u\in V^{(\cQ(\cH))}_{\al}$ and $L(u)\subseteq\cQ$.  
In particular, $u\in\VL$ if and only if
$u\in\VQH$ and $L(u)\subseteq\cQ$. 
Moreover, $\rank(u)$ is the least $\al$ such 
that $u\in \VQH_{\al}$ for  any $u\in\VL$.
\end{Proposition}
\begin{Proof}Immediate from transfinite induction on
$\al$.
\end{Proof}

Let $\cA\subseteq\VQ$.  The {\em commutator
of $\cA$}, denoted by $\cm(\cA)$, is defined by 
\[
\cuniv(\cA)=\com (L(\cA)).
\]
For any $u_1,\ldots,u_n\in\VQ$, we write
$\cuniv(u_1,\ldots,u_n)=\cuniv(\{u_1,\ldots,u_n\})$.

Let $u\in\VQ$ and $p\in\cQ$.
The {\em restriction} $u|_p$ of $u$ to $p$ is defined by
the following transfinite recursion:
\beqas
\dom(u|_p)&=&\{x|_p\mid x\in\dom(u)\},\\
u|_p(x|_p)&=&u(x)\And p
\eeqas
for any $x\in\dom(u)$.
By induction, it is easy to see that if $q, p\in\cQ$, then $(u|_p)|_q=u|_{p\And q}$ for 
all $u\in\VQ$.

\begin{Proposition}\label{th:L-restriction}
For any $\cA\subseteq \VQ$ and $p\in\cQ$, 
we have 
\[
L(\{u|_p\mid u\in\cA\})=L(\cA)\And p.
\]
\end{Proposition}
\begin{Proof}
By induction, it is  easy to see the relation
$
L(u|_p)=L(u)\And p,
$
so that the assertion follows easily.
\end{Proof}

Let $\cA\subseteq\VQ$.  The {\em logic
generated by $\cA$}, denoted by $\cQ(\cA)$, is  define by 
\[
\cQ(\cA)=L(\cA)^{!!}.
\]
For $u_1,\ldots,u_n\in\VQ$, we write
$\cQ(u_1,\ldots,u_n)=\cQ(\{u_1,\ldots,u_n\})$.

\begin{Proposition}\label{th:range}
For any $\De_0$-formula $\ph(x_1,\ldots,x_n)$ in
$\LL(\in)$ and $u_1,\cdots,u_n\in\VQ$,
we have $\val{\ph(u_1,\ldots,u_n)}\in\cQ(u_1,\ldots,u_n)$.
\end{Proposition}
\begin{Proof}
Let $\cA=\{u_1,\ldots,u_n\}$.
Since $L(\cA)\subseteq\cQ(\cA)$, it follows from
Proposition \ref{th:sublogic} that $u_1,\ldots,u_n\in
V^{(\cQ(\cA))}$.
By the $\De_0$-absoluteness
principle, we have 
$\val{\ph(u_1,\ldots,u_n)}=
\val{\ph(u_1,\ldots,u_n)}{}_{\cQ(\cA)}\in \cQ(\cA)$.
\end{Proof}

\begin{Proposition}\label{th:commutativity}
For any 
$\De_{0}$-formula ${\ph} (x_{1},{\ldots}, x_{n})$ 
in $\LL(\in)$ and $u_{1},{\ldots}, u_{n}\in\VQ$, if 
$p\in L(u_1,\ldots,u_n)^{!}$, then 
$p\commutes \val{\ph(u_1,\ldots,u_n)}$
and $p\commutes \val{\ph(u_1|_p,\ldots,u_n|_p)}$.
\end{Proposition}
\begin{Proof}
Let $u_{1},{\ldots}, u_{n}\in\VQ$.
If $p\in L(u_1,\ldots,u_n)^{!}$, then
$p\in \cQ(u_1,\ldots,u_n)^{!}$.  From Proposition 
\ref{th:range},
$\val{\ph(u_1,\ldots,u_n)}\in\cQ(u_1,\ldots,u_n)$,
so that $p\commutes \val{\ph(u_1,\ldots,u_n)}$.
From Proposition \ref{th:L-restriction},
$L(u_1|_p,\ldots,u_n|_p)=L(u_1,\ldots,u_n)\And p$,
and hence $p\in L(u_1|_p,\ldots,u_n|_p)^{!}$, so that
$p\commutes \val{\ph(u_1|_p,\ldots,u_n|_p)}$.
\end{Proof}

We define the binary relation $x_1\subseteq x_2$ by
$\forall x\in x_1(x\in x_2)$.
Then, by definition for  any $u,v\in\VQ$ we have
\[
\val{u\subseteq v}=
\Inf_{u'\in\dom(u)}
u(u')\Then \val{u'\in v},
\]
and we have $\val{u=v}=\val{u\subseteq v}
\And\val{v\subseteq u}$.

\begin{Proposition}\label{th:restriction-atom}
For any $u,v\in\VQ$ and $p\in L(u,v)^{!}$, we have
the following relations.

(i) $\val{u|_p\in v|_p}=\val{u\in v}\And p$.

(ii) $\val{u|_p\subseteq v|_p}\And p
=\val{u\subseteq p}\And p$.

(iii) $\val{u|_p= v|_p}\And p =\val{u= p}\And p$
\end{Proposition}
\begin{Proof}
We prove the relations by induction on the ranks of 
$u,v$.  If $\rank(u)=\rank(v)=0$, then $\dom(u)=\dom(v)
=\emptyset$, so that the relations trivially hold.
Let $u,v\in\VQ$ and $p\in L(u,v)^{!}$.
To prove (i),
let $v'\in\dom(v)$. 
Then, we have $p\commutes v(v')$ by the assumption on $p$.
By induction hypothesis, we have also 
$\val{u|_p=v'|_p}\And p=\val{u=v'}\And p$.
By Proposition \ref{th:commutativity}, we have 
$p\commutes \val{u=v'}$, so that
$v(v'), \val{u=v'}\in\{p\}^{!}$, and hence
$v(v')\And\val{u=v'}\in\{p\}^{!}$. 
Thus,  we  have
\beqas
\val{u|_p\in v|_p}
&=&\Sup_{v'\in\dom(v|_p)}
v|_p(v')\And\val{u|_p=v'}\\
&=&
\Sup_{v'\in\dom(v)}
v|_p(v'|_p)\And\val{u|_p=v'|_p}\\
&=&
\Sup_{v'\in\dom(v)}
(v(v')\And p)\And(\val{u=v'}\And p)\\
&=&
\left(\Sup_{v'\in\dom(v)}
(v(v')\And \val{u=v'})\And p\right)\\
&=&
\left(\Sup_{v'\in\dom(v)}
v(v')\And \val{u=v'}\right)\And p,
\eeqas  
where the last equality follows from Proposition \ref{th:logic}.
Thus, by definition of $\val{u=v}$ we obtain
the relation $\val{u|_p\in v|_p}=\val{u=v}\And p$,
and relation (i) has been proved.
To prove (ii), let $u'\in\dom(u)$.
Then, we have $\val{u'|_p\in v|_p}=\val{u'\in v}\And p$
by induction hypothesis.
Thus, we have
\beqas
\val{u|_p\subseteq v|_p}
&=&
\Inf_{u'\in\dom(u|_p)}(u|_p(u')\Then\val{u'\in v|_p})\\
&=&
\Inf_{u'\in\dom(u)}(u|_p(u'|_p)\Then\val{u'|_p\in v|_p})\\
&=&
\Inf_{u'\in\dom(u)}
(u(u')\And p)\Then(\val{u'\in v}\And p).
\eeqas
We  have $p\commutes u(u')$ by
assumption on $p$, and $p\commutes\val{u'\in v}$
by Proposition \ref{th:commutativity},
so that
$p\commutes u(u')\Then\val{u'\in v}$ and
$p\commutes (u(u')\And p)\Then(\val{u'\in v}\And p)$.
Thus, by Proposition \ref{th:logic} and Proposition \ref{th:GC} (ii) we have
\beqas
p\And\val{u|_p\subseteq v|_p}
&=&
p\And\Inf_{u'\in\dom(u)}
(u(u')\And p)\Then(\val{u'\in v}\And p)\\
&=&
p\And\Inf_{u'\in\dom(u)}(u(u')\Then\val{u'\in v})\\
&=&
p\And\val{u\subseteq v}.
\eeqas
Thus, we have proved  relation (ii).  
Relation (iii) follows easily from relation (ii).
\end{Proof}

We have the following theorem.
\begin{Theorem}[$\De_0$-Restriction Principle]
\label{th:V-restriction}
For any $\De_{0}$-formula ${\ph} (x_{1},{\ldots}, x_{n})$ 
in $\LL(\in)$ and $u_{1},{\ldots}, u_{n}\in\VQ$, if 
$p\in L(u_1,\ldots,u_n)^{!}$, then 
$\val{\ph(u_1,\ldots,u_n)}\And p=
\val{\ph(u_1|_p,\ldots,u_n|_p)}\And p$.
\end{Theorem}
\begin{Proof}
We prove the assertion by induction on 
the complexity of  ${\ph} (x_{1},{\ldots},x_{n})$.
From Proposition \ref{th:restriction-atom}, the assertion
holds for atomic formulas.
Then, the verification of every induction step follows 
from the fact that (i) the function $a\mapsto a\And p$ of all $a\in
\{p\}^{!}$ preserves the supremum and the infimum
as shown in Proposition \ref{th:logic},
 (ii) it satisfies $(a\Then b)\And p=[(a\And p)\Then (b\And p)] 
\And p$ for  all $a,b\in\{p\}^{!}$
from the defining property of generalized implications,
(iii) it satisfies relation (ii) of Theorem \ref{th:implication},
and that (iv) it satisfies the relation
$a^{\perp}\And p=(a\And p)^{\perp}\And p$
for  all $a,b\in\{p\}^{!}$.
\end{Proof}

Now,  we obtain the following transfer principle for bounded theorems of ZFC
with respect to any material conditionals $\Then$,
which is of the same form as Theorem 4.6 of Ref.~\cite{07TPQ} obtained for 
the Sasaki conditional $\ThenS$.

\begin{Theorem}[$\De_{0}$-ZFC  Transfer Principle]
For any $\De_{0}$-formula ${\ph} (x_{1},{\ldots}, x_{n})$ 
of $\cL(\in)$ and $u_{1},{\ldots}, u_{n}\in\VQ$, if 
${\ph} (x_{1},{\ldots}, x_{n})$ is provable in ZFC, then
we have
\[
\cuniv(u_{1},\ldots,u_{n})\le
\val{\ph({u}_{1},\ldots,{u}_{n})}.
\]
\end{Theorem}
\begin{Proof}
Let $p=\cuniv(u_1,\ldots,u_n)$.
Then, we have $a\And p\commutes b\And p$
for any $a,b\in L(u_1,\ldots,u_n)$, and hence 
there is a Boolean sublogic $\cB$ such that 
$L(u_1,\ldots,u_n)\And p\subseteq \cB$.
From Proposition \ref{th:L-restriction},
we have $L(u_1|_p,\ldots,u_n|_p)\subseteq \cB$.
From Proposition \ref{th:sublogic}, we have
$u_1|_p,\ldots,u_n|_p\in \VB$.
By the ZFC transfer principle of the Boolean-valued
universe \cite[Theorem 1.33]{Bel85}, we have
$\val{\ph(u_1|_p,\ldots,u_n|_p)}{}_{\cB}=1$. By the
$\De_0$-absoluteness principle, we have
$\val{\ph(u_1|_p,\ldots,u_n|_p)}=1$.
From Proposition \ref{th:V-restriction}, we have
$\val{\ph(u_1,\ldots,u_n)}\And p
=\val{\ph(u_1|_p,\ldots,u_n|_p)}\And p
=p$, and the assertion follows.
\end{Proof}

\section{Real numbers in quantum set theory}
\label{se:RN}
\renewcommand{\Then}{\rightarrow}
\renewcommand{\vval}[1]{[\![#1]\!]}
\renewcommand{\valj}[1]{[\![#1]\!]_j}
\newcommand{\bproof}{\begin{proof}}
\newcommand{\eproof}{\end{proof}}

Let $\Q$ be the set of rational numbers in $V$.
We define the set of rational numbers in the model $\VQ$
to be $\check{\Q}$.
We define a real number in the model by a Dedekind cut
of the rational numbers. More precisely, we identify
a real number with the upper segment of a Dedekind cut
assuming that the lower segment has no end point.
Therefore, the formal definition of  the predicate $\R(x)$, 
``$x$ is a real number,'' is expressed by
\beqa
\R(x)&:=&
\forall y\in x(y\in\check{\Q})
\And \exists y\in\check{\Q}(y\in x)
\And \exists y\in\check{\Q}(y\not\in x)\nn\\
& &  \And
\forall y\in\check{\Q}(y\in x\Iff\forall z\in\check{\Q}
(y<z \Then z\in x)).
\eeqa
We define $\R_j^{(\cQ)}$ for $j=S,C,R$ to be the interpretation of 
the set $\R$ of real numbers in $\VQ$ under the $j$-conditional as follows.
\[
\R_j^{(\cQ)} = \{u\in\VQ|\ \dom(u)=\dom(\check{\Q})
\ \mb{and }\val{\R(u)}=1\}.
\]
The set $\R_{j,\cQ}$ of real numbers in $\VQ$ under the $j$-conditional
is defined by
\beq
\R_{j,\cQ}=\R_j^{(\cQ)}\times\{1\}.
\eeq
\begin{Proposition}
\begin{enumerate}[{\rm (i)}]\itemsep=0in
\item For any $u\in\VQ$ with $\dom(u)=\dom(\check{\Q})$, we have
\[
\displaystyle \vval{\R(u)}_j
=\Sup_{y\in\Q} u(\ck{y})
\And
\left(\Inf_{y\in\Q}u(\ck{y})\right)^\perp
\And
\Inf_{y\in\Q}\left(u(\ck{y})\Iff\Inf_{y<z} u(\ck{z})\right).
\]
\item 
$
\R_{S,\cQ}=\R_{C,\cQ}=\R_{R,\cQ}.
$
\eenum
\end{Proposition}
\begin{Proof}
Assertion (i) follows from the definition.
Assertion (ii) holds, since  $\vval{\R(u)}_j$ is independent of the
choice of conditional.
\end{Proof}

From the above, in what follows we will write 
$\R^{(\cQ)}=\R_j^{(\cQ)}$ and 
$\R_{\cQ}=\R_{j,\cQ}$.

\begin{Theorem}
\label{th:RQ}
For any $u\in\R^{(\cQ)}$, we have the following.

(i) $u(\check{r})=\vval{\check{r}\in u}_j$ for all $r\in\Q$ and $j=S,R,C$.

(ii)  $\cuniv(u)=1$.
\end{Theorem}
\begin{Proof}
Let $u\in\R^{(\cQ)}$ and $r\in\Q$.
Then, we have
\[
\val{\check{r}\in u}
= \Sup_{s\in\Q}(\val{\check{r}=\check{s}}\And u(\check{s}))\\
= u(\check{r}),
\]
since $\val{\check{r}=\check{s}}=1$ if $r=s$ and  $\val{\check{r}=\check{s}}=0$
otherwise by the $\De_0$-Elementary Equivalence Principle,
and assertion (i) follows.
We have
\beqas
L(u)
=\bigcup_{s\in\Q}L(\check{s})
\cup
\{u(\check{s})\mid s\in\Q\}
=\{0,1,u(\check{s})\mid s\in\Q\},
\eeqas
so that it suffices to show that each $u(\check{s})$
with $s\in\Q$ is mutually commuting. 
By definition, we have
\[
\val{\forall y\in\check{\Q}(y\in u\Iff\forall z\in\check{\Q}
(y<z \Then z\in u))}
=1.
\]
Hence, we have
\[
u(\check{s})=\Inf_{s<t, t\in\Q}u(\check{t}).
\]
Thus, if $s_1<s_2$, then
$u(\check{s}_1)\le u(\check{s}_2)$,
so that
$u(\check{s}_1)\commutes u(\check{s}_2)$.
\end{Proof}

Let $\cM$ be a von Neumann algebra on a Hilbert
space $\cH$ and let $\cQ=\cP(\cM)$.
A closed operator $A$ (densely defined) on $\cH$ is
said to be {\em affiliated} with $\cM$, in symbols $A\,\et\,\cM$, 
iff $U^{*}AU=A$ for any unitary operator $U\in\cM'$.
Let $A$ be a self-adjoint operator (densely defined) on $\cH$
and let $A=\int_{\R} \la\, dE^A(\la)$
be its spectral decomposition, where $\{E^{A}(\la)\}_{\la\in\R}$ 
is the resolution of  identity belonging to $A$ \cite[p.\ 119]{vN55}.
It is well-known that $A\af\cM$ if and only if $E^A({\la})\in\cQ$ 
for every $\la\in\R$.
Denote by $\overline{\cM}_{SA}$ the set of self-adjoint operators 
affiliated with $\cM$.
Two self-adjoint operators $A$ and $B$ are said to {\em commute},
in symbols $A\commutes B$,
iff $E^A(\la)\commutes E^B(\la')$ for every pair 
$\la,\la'$ of reals.

For any $u\in\R^{(\cQ)}$ and $\la\in\R$, we define $E^{u}(\la)$ by 
\beq
E^{u}(\la)=\Inf_{\la<r\in\Q}u(\check{r}).
\eeq
Then, it can be shown that
$\{E^u(\la)\}_{\la\in\R}$ is a resolution of
identity in $\cQ$ and hence by the spectral theorem there
is a self-ajoint operator $\hat{u}\af\cM$ uniquely
satisfying $\hat{u}=\int_{\R}\la\, dE^u(\la)$.  On the other
hand, let $A\af\cM$ be a self-ajoint operator. We define $\tilde{A}\in\VQ$ by
\beq
\tilde{A}=\{(\check{r},E^{A}(r))\mid r\in\Q\}.
\eeq
Then, $\dom(\tA)=\dom(\check{\Q})$ and $\tA(\check{r})=E^{A}(r)$ for all $r\in\Q$.
It is easy to see that $\tilde{A}\in\RQ$ and we have
$(\hat{u})\tilde{}=u$ for all $u\in\RQ$ and $(\tilde{A})\hat{}=A$
for all $A\in\overline{\cM}_{SA}$.
Therefore, the correspondence
between $\RQ$ and $\overline{\cM}_{SA}$ is a one-to-one correspondence.
We call the above correspondence the {\em Takeuti correspondence}.
Now, we have the following \cite[Theorem 6.1]{07TPQ}.

\begin{Theorem}
Let $\cQ$ be a logic on $\cH$.  The relations 
\bitem
\item[\rm (i)] ${\displaystyle E^{A}(\la)=\Inf_{\la<r\in\Q}u(\check{r})}$ for all $\la\in\Q$,
\item[\rm (ii)] $u(\check{r})=E^{A}(r)$ for all $r\in\Q$,
\eitem
for all $u=\tA\in\RQ$ and $A=\hu\in \overline{\cM}_{SA}$ 
sets up a one-to-one correspondence between $\RQ$ and $ \overline{\cM}_{SA}$.
\end{Theorem}

For any $r\in\R$, we shall write $\tilde{r}=(r1)\,\tilde{}$, where $r1$ is the scalar 
operator on $\cH$.
Then, we have $\dom(\tilde{r})=\dom(\check{\Q})$ and $\tilde{r}(\check{t})=
\val{\check{r}\le \check{t}}$, so that we have $\L(\tilde{r})=\{0,1\}$.
Denote by $\cB(\R^n)$ the $\si$-filed of Borel subsets of $\R^n$ and $B(\R^n)$ the space
of bounded Borel functions on $\R^n$.
For any $f\in B(\R)$, the bounded self-adjoint operator $f(X)\in\cM$ is
defined by $f(X)=\int_{\R}f(\la) dE^{X}(\la)$.
For any Borel subset $\De$ in $\R$, we denote by $E^{X}(\De)$ the spectral
projection corresponding to $\De\in\cB(\R)$, {\em i.e., } 
$E^{X}(\De)=\ch_{\De}(X)$, where $\ch_{\De}$ is the characteristic function of $\De$.  
Then, we have $E^{X}(\la)=E^{X}((-\infty,\la])$.  
The following proposition is a straightforward consequence of definitions. 

\begin{Proposition}\label{th:QBorel}
Let $r\in\R$, $s,t\in\R$, and $X\af\cM_{SA}$.
For $j=R,C,S$, we have the following relations.
\bitem
\item[\rm (i)] $\val{\check{r}\in\tilde{s}}=\val{\check{s}\le \check{r}}
=E^{s1}(t)$.
\item[\rm (ii)] $\val{\tilde{s}\le\tilde{t}}=\val{\check{s}\le\check{t}}
=E^{s1}(t)$.
\item[\rm (iii)] $\val{\tilde{X}\le\tilde{t}}=E^{X}(t)=E^{X}((-\infty,t])$.
\item[\rm (iv)] $\val{\tilde{t}<\tilde{X}}=1-E^{X}(t)=E^{X}((t,\infty))$.
\item[\rm (v)] $\val{\tilde{s}<\tilde{X}\le \tilde{t}}=E^X({t})-E^X({s})=
E^{X}((s,t])$.
\item[\rm (vi)] $\val{\tilde{X}=\tilde{t}}
=E^{X}(t)-\Sup_{r<t,r\in\Q} E^{X}(r)=E^{X}(\{t\})$.
\eitem
\end{Proposition}

In what follows, we write $r\And s=\min\{r,s\}$
and $r\Or s=\max\{r,s\}$ for any $r,s\in\R$.

The $\cQ$-value of equality  $\val{u=v}$ for $u,v\in\RQ$ is 
independent of the choice of the conditional and characterized
as follows.

\begin{Theorem}
\label{th:equality}
For any $u,v\in\R^{(\cQ)}$ we have
\[
 \val{u=v}=\cP\{\ps\in\cH\mid
 u(\check{x})\ps=v(\check{x})\ps
\mbox{ for all }x\in \Q\}.
\]
\end{Theorem}
\begin{Proof}
From Theorem \ref{th:RQ} (i) we have
\beqas
\val{u=v}
&=&
\Inf_{r\in \Q}(u(\check{r})\Then\val{\check{r}\in v})\And
\Inf_{r\in \Q}(v(\check{r})\Then\val{\check{r}\in u})\\
&=&
\Inf_{r\in \Q}(u(\check{r})\Iff v(\check{r})).
\eeqas
From Proposition \ref{th:com} (iv), we have
\[
u(\check{r})\Iff v(\check{r})
=
\cP\{\ps\in\cH\mid u(\check{r})\ps=
v(\check{r})\ps\}
\]
Thus, the assertion follows easily.
\end{Proof}

\begin{Theorem}\label{th:forcing_real}
For any $u,v\in\R^{(\cQ)}$ 
and $\ps\in\cH$, the following conditions are all
equivalent.

(i)  $\ps\in\cR\val{u=v}$.

(ii) $u(\check{x})\ps=v(\check{x})\ps$ for any $x\in\Q$.

(iii) $u(\check{x})v(\check{y})\ps=v(\check{x}\And\cy)\ps$ for any
$x,y\in\Q$.

(iv)
$\bracket{u(\check{x})\ps,v(\check{y})\ps}
=\|v(\check{x}\And\check{y})\ps\|^{2}$
for any $x,y\in\Q$.

\end{Theorem}
\begin{Proof}
The equivalence (i) $\IFF$ (ii) follows from Theorem \ref{th:equality}.
Suppose (ii) holds.
Then, we have
$
u(\cx)v(\cy)\ps=u(\cx)u(\cy)\ps=u(\cx\And \cy)\ps=v(\cx\And \cy)\ps.
$
Thus, the implication (ii) $\THEN$ (iii) holds.
Suppose (iii) holds.
We have 
$
\bracket{u(\cx)\ps,v(\cy)\ps}
=\bracket{\ps,u(\cx)v(\cy)\ps}
=\bracket{\ps,v(\cx\And\cy)\ps}
=\|v(\cx\And\cy)\ps\|^{2},
$
and hence the implication (iii)$\THEN$(iv) holds.
Suppose (iv) holds.
Then, we have $\bracket{u(\cx)\ps,v(\cx)\ps}=\|v(\cx)\ps\|^{2}$ and
$\bracket{v(\cx)\ps,u(\cx)\ps}=\|u(\cx)\ps\|^2$.  Consequently, we have 
$
\|u(\cx)\ps-v(\cx)\ps\|^{2}
=
\|u(\cx)\ps\|^2+\|v(\cx)\ps\|^{2}-\bracket{u(\cx)\ps,v(\cx)\ps}
-\bracket{v(\cx)\ps,u(\cx)\ps}
=0, 
$
and hence $u(\cx)\ps=v(\cx)\ps$.  Thus, the implication (iv)$\THEN$(ii)
holds, and the proof is completed.
\end{Proof}

The set $\R_\cQ$ of real numbers in $\VQ$ is defined by
\beqas
\R_\cQ=\R^{(\cQ)}\times\{1\}.
\eeqas

Let $A$ be an observable.
For any (complex-valued) bounded Borel function $f$ on $\R$, 
we define the observable
$f(A)$ by
\[
f(A)=\int_{\R}f(\la)\,dE^{A}(\la).
\]
We shall denote by $B(\R)$ the space of bounded Borel functions on $\R$.
For any Borel set $\De$ in $\R$, we define $E^{A}(\De)$ by
$E^{A}(\De)=\ch_{\De}(A)$, where $\ch_{\De}$ is a Borel function
on $\R$ defined by $\ch_{\De}(x)=1$ if $x\in\De$ and $\ch_{\De}(x)=0$ if
$x\not\in\De$. 
For any pair of observables $A$ and $B$, the {\em joint probability distribution}
of $A$ and $B$ in a state $\ps$ is a probability measure $\mu^{A,B}_{\ps}$
on $\R^{2}$ satisfying 
\[
\mu^{A,B}_{\ps}(\De\times\Ga)=
\bracket{\ps,(E^{A}(\De)\And E^{B}(\Ga))\ps}
\]
for any $\De,\Ga\in\cB(\R)$.
Gudder \cite{Gud68} showed that the joint probability distribution $\mu^{A,B}_{\ps}$
exists if and only if the relation $[E^{A}(\De),E^{B}(\Ga)]\ps=0$ holds for
every $\De,\Ga\in\cB(\R)$.

\begin{Theorem}\label{th:qpc}
For any observables
(self-adjoint operators) $A, B$ on $\cH$ and any state
(unit vector) $\ps\in\cH$, the following conditions are all equivalent.

(i)  $\ps\in\cR\val{\tilde{A}=\tilde{B}}$.

(ii) $E^{A}(r)\ps=E^{B}(r)\ps$ for any $r\in\Q$.

(iii) $f(A)\ps=f(B)\ps$ for all $f\in B(\R)$.

(iv) $\bracket{E^{A}(\De)\ps,E^{B}(\Ga)\ps}=0$
for any $\De,\Ga\in\cB(\R)$ with $\De\cap\Ga=\emptyset$.

(v) There is the joint probability distribution $\mu^{A,B}_{\ps}$ of $A$ and $B$
in $\ps$ satisfying 
\[
\mu^{A,B}_{\ps}(\{(a,b)\in\R^{2}\mid a=b\})=1.
\]

\end{Theorem}
\begin{Proof}
The equivalence (i) $\IFF$ (ii)
follows from Theorem \ref{th:forcing_real}.
Suppose that (ii) holds.  Let $\la\in\R$.  
If $r_1,r_2,\ldots$ be a decreasing sequence of rational numbers
convergent to $\la$, then $E^{A}(r_n)\ps$ and $E^{B}(r_n)\psi$ 
are convergent to $E^{A}(\la)\ps$ and $E^{B}(\la)\ps$, respectively, so that
$E^{A}(\la)\ps=E^{B}(\la)\ps$ for all $\la\in\R$.  Thus, we have 
$$
\bracket{\xi,f(A)\ps}=\int_\R f(\la)\,d\bracket{\xi,E^{A}(\la)\ps}=\int_\R
f(\la)\,d\bracket{\xi,E^{B}(\la)\ps}=\bracket{\xi,f(B)\ps}
$$ 
for all $\xi\in\cH$, and hence we have $f(A)\ps=f(B)\ps$ for all
$f\in B(\R)$.   Thus, the implication (ii) $\THEN$ (iii) holds.
Since condition (ii) is a special case of condition (iii) where $f=\ch_{(-\infty,r]}$,
the implication (iii) $\THEN$ (ii) is trivial, so that the equivalence
(ii) $\IFF$ (iii) follows.  The equivalence of assertions (iii), (iv), and (v) have
been already proved in Ref.~\cite{06QPC}, the proof is completed.
\end{Proof}

Condition (iii) above is adopted as the defining condition for $A$ and $B$ to be
perfectly correlated in $\ps$ ~\cite{06QPC} because of the simplicity of the formulation.
Condition (v) justifies our nomenclature calling $A$ and $B$ ``perfectly 
correlated.''
By condition (i), quantum logic justifies the assertion that ``perfectly 
correlated'' observables actually have the same value in the given state.
For further properties and applications of the notion of perfect correlation,
we refer the reader to Ref.~\cite{06QPC}.

The following converse statement of the $\De_{0}$-ZFC  Transfer Principle
can be proved using the interpretation of real numbers in $\VQ$.

\begin{Theorem}[Converse of the $\De_{0}$-ZFC  Transfer Principle]
Let $j=S,C,R$.
If 
\[
\val{\ph({u}_{1},\ldots,{u}_{n})}=1
\]
holds for any $\De_{0}$-formula ${\ph} (x_{1},{\ldots}, x_{n})$ 
of $\cL(\in)$ provable in ZFC
and  $u_{1},{\ldots}, u_{n}\in\VQ$,  then the logic $\cQ$ is Boolean.
\end{Theorem}
\begin{Proof}
Let $P,Q\in\cQ$. 
Then, by the Takeuti correspondence we have $\tP,\tQ$ in $\RQ$.
By Proposition \ref{th:QBorel}, we have 
$\val{\check{0}\in\tilde \tP}=\tP(\check{0})=\val{\tP\le \tilde{0}}=E^{P}(0)=I-P$,  $\val{\check{0}\not\in\tilde \tP}
=P$,  $\val{\check{0}\in\tilde \tQ}=I-Q$, and $\val{\check{0}\not\in\tilde \tQ}=Q$.
Since 
\[
z\in x \Iff [(z\in x \And z\in y)\Or(z\in x \And z\not\in y)]
\] 
is provable in ZFC, by assumption we have
\[
\val{\check{0}\in \tP}= (\val{\check{0}\in \tP} \And \val{\check{0}\in \tQ})\Or(\val{\check{0}\in \tP} 
\And \val{\check{0}\not\in \tQ}),
\]
so that $P=(P\And Q)\Or(P\And Q^{\perp})$, and $P\commutes Q$.
Since  $P,Q\in\cQ$ were arbitrary, we conclude $\cQ$ is Boolean. 
\end{Proof}

\section{Order relations on quantum reals}
\label{se:7}

\newcommand{\les}{\preccurlyeq}
\newcommand{\ges}{\succcurlyeq}

Since the real numbers are defined as the upper segment of Dedekind cuts
of rational numbers whose lower segment has no end point,
the order relation between two quantum reals $u,v\in\RQ$ is defined as
\[
u\le v:=(\forall r\in v)[r\in u].
\]
Let $\cM=\cQ''$.
For any self-adjoint operators $X,Y\af \cM$ we write $X\les Y$ iff 
$E^{Y}(\la)\le E^{X}(\la)$ for all $\la\in\R$.  The relation is called the 
{\em spectral order}.  
This order is originally introduced by Olson \cite{Ols71} for bounded operators; 
for recent results for unbounded operators see \cite{PS12}.
With the spectral order the set $\overline{\cM}_{SA}$ is a conditionally complete
lattice, but it is not a vector lattice; in contrast to the fact that the usual linear order 
$\le$ of self-adjoint operators is a lattice if and only if $\cM$ is abelian.  
The following facts about the spectral order are known \cite{Ols71,PS12}:
\begin{enumerate}[(i)]\itemsep=0in
\item The spectral order coincides with the usual linear order on projections
and mutually commuting operators.
\item For any $0\le X,Y\af\cM_{SA}$, we have $X\les Y$ if and only if $X^{n}\le Y^{n}$
for all $n\in\N$.
\eenum

\begin{Proposition}
For any $X,Y\af\cM_{SA}$ and $j=S,C,R$, we have $\vval{\tX \le \tY}_j=1$ if and only if
 $X\les Y$.
\end{Proposition}

\begin{proof}
We have
\[
\vval{\tX \le \tY}_j=\val{(\forall r\in \tY)[r\in \tX]}
=\Inf_{r\in\scriptsize\dom(\tY)}\tY(r)\Thenjj \val{r\in\tX}
=\Inf_{r\in\Q}E^{Y}(r)\Thenjj E^{X}(r)
\]
Thus, the assertion follows from the fact that $E^{Y}(r)\le E^{X}(r)$
if and only if $E^{Y}(r)\Thenjj E^{X}(r)=1$.
\end{proof}

To clarify the operational meaning of the truth value $\val{\tX\le \tY}$, in what follows
we shall confine our attention to the case where $\cH$ is finite dimensional.

Let $X=\sum_{k=1}^{n}x_n E^{X}(\{x_n\})$ and $Y=\sum_{k=1}^{m}y_m E^{X}(\{x_n\})$
be the spectral decomposition of $X$ and $Y$, where $x_1<\cdots<x_n$ and $y_1<\cdots<y_n$.
Then, we have 
\beqas
E^{X}(x)&=&\sum_{k:x_k\le x} E^{X}(\{x_k\}),\\
E^{Y}(y)&=&\sum_{k:y_k\le y} E^{Y}(\{y_k\}).\\
\eeqas

We define the joint probability distribution $P^{X,Y}_{\psi}(x,y)$ 
representing the joint probability of obtaining the outcomes $Y=y$ and $X=x$ 
from the successive projective measurements of $Y$ and $X$, where the $X$ 
measurement follows immediately after the $Y$ measurement on the same system
prepared in the state $\psi$ just before the $Y$ measurement (see Figure \ref{fig:one}).
Then, it is well-know that $P^{X,Y}_{\psi}(x,y)$ is determined by
\[
P^{X,Y}_{\psi}(x,y)=\| E^{X}(\{x\})E^{Y}(\{y\})\psi\|^2.
\]
Analogously, we define the joint probability distribution $P^{Y,X}_{\psi}(y,x)$ obtained
by the projective $X$ measurement and the immediately following $Y$ measurement
(see Figure \ref{fig:one}).  Then,
we have
\[
P^{Y,X}_{\psi}(y,x)=\| E^{Y}(\{y\})E^{X}(\{x\})\psi\|^2.
\]

\begin{figure}[h]
  \begin{center}
\includegraphics[width=12cm]{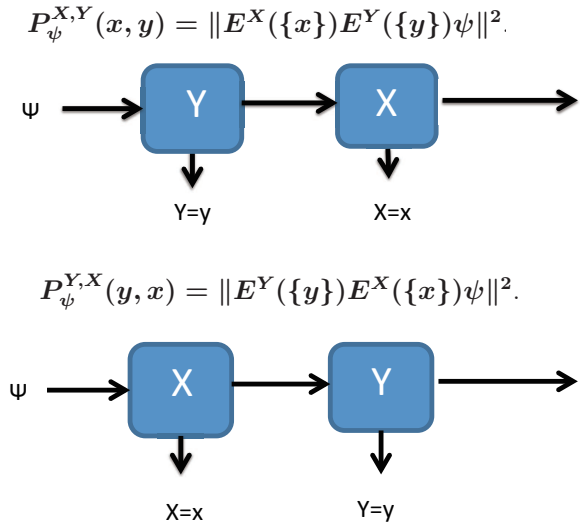}
\caption{Successive projective measurements}
 \label{fig:one}
  \end{center}
\end{figure}

Then, we have the following.

\begin{Theorem}
Let $X$ and $Y$ be observables on a finite dimensional Hilbert space $\cH$
and $\psi$ be a state in $\cH$.   Then, we have the following.
\begin{enumerate}[\rm (i)]\itemsep=0in
\item $\psi\in\cR(\vval{\tX\le \tY}_S)$ if and only if 
$P^{X,Y}_{\psi}(x,y)=0$ for any $x,y\in\R$ such that $x>y$.
\item $\psi\in\cR(\vval{\tX\le \tY}_C)$ if and only if 
$P^{Y,X}_{\psi}(y,x)=0$ for any $x,y\in\R$ such that $x>y$.
\item $\psi\in\cR(\vval{\tX\le \tY}_R)$ if and only if 
$P^{Y,X}_{\psi}(y,x)= P^{X,Y}_{\psi}(x,y)=0$ for any
$x,y\in\R$ such that $x>y$.
\eenum
\end{Theorem}
\bproof
Let $\psi\in\cR(\vval{\tX\le\tY}_S)$.  From Theorem \ref{th:com} we have
$E^{X}(\la)^\perp E^{Y}(\la)\psi=0$ for any $\la\in\R$.
Now we shall show the relation
\beql{r}
E^{X}(\la)^\perp E^{Y}(\{\la\})\psi= 0
\eeq
for any $\la\in\R$.
If $\la$ is not an eigenvalue of $Y$, we have $E^{Y}(\{\la\})=0$
and relation \eq{r} follows.  Suppose $\la=y_k$.  If $k=1$, then 
$E^{Y}(\la)=E^{Y}(\{\la\})$ and hence relation \eq{r} follows.
By induction we assume $E^{X}(y_j)^\perp E^{Y}(\{y_j\})\psi= 0$ for all $j<k$.
Since $E^{X}(y_k)^{\perp}E^{X}(y_j)^\perp=E^{X}(y_k)^{\perp}$, we have
$E^{X}(y_k)^\perp E^{Y}(\{y_j\})\psi= 0$ for all $j<k$.
Thus, we have $E^{X}(y_k)^\perp E^{Y}(y_{k-1})\psi=\sum_{j=1}^{k-1}
E^{X}(y_k)^\perp E^{Y}(\{y_j\})\psi= 0$.
It follows that $E^{X}(\la)^\perp E^{Y}(\{\la\})\psi=E^{X}(\la)^\perp E^{Y}(\la)\psi
-E^{X}(y_k)^\perp E^{Y}(\{y_j\})\psi=0$.
Thus, relation \eq{r} holds for any $\la\in\R$.
Thus, if $x>y$ then we have 
$P^{X,Y}_{\psi}(x,y)=\|E^{X}(\{x\}) E^{Y}(\{y\})\psi\|^2= 0$.
Conversely, suppose that the last equation holds.  
Then, we have $E^{X}(\{x\})E^{Y}(\{y\})\psi=0$ for all
$x>y$, so that it easily follows that 
$E^{X}(\la)^{\perp}E^{Y}(\la)\psi=0$ for every $\la\in\R$.
Thus, assertion (i) follows from Theorem \ref{th:com}.
The rest of the assertions follow routinely.
\eproof

Note that $P^{X,Y}_{\psi}(x,y)=0$ for any $x,y\in\R$ such that $x>y$ if and only
if $\sum_{x\le y}P^{X,Y}_{\psi}(x,y)=1$ if and only if the outcome of the 
$X$-measurement is less than or equal to the outcome of the $Y$-measurement 
in a successive $(Y,X)$-measurement with probability 1.
Similarly, $P^{Y,X}_{\psi}(y,x)=0$ for any $x,y\in\R$ such that $x>y$ if and only if
$\sum_{x\le y}P^{Y,X}_{\psi}(y,x)=1$
if and only if
the outcome of the $X$-measurement is less than or equal to the outcome of the $Y$-measurement 
in a successive $(X,Y)$-measurement with probability 1.

\section{Conclusion}

In quantum logic there are at least three candidates for conditional operation,
called the Sasaki conditional, the contrapositive Sasaki conditional,
and the relevance conditional.  In this paper, we have attempted to
develop quantum set theory based on quantum logics with those three 
conditionals, each of which defines different quantum logical truth
value assignment.  We have shown that those three models 
satisfy the transfer principle of the same form to determine the quantum 
logical truth values of theorems of the ZFC set theory.
We also show that the reals in the model and the truth values of their equality 
are the same for those models.
Interestingly, however, we have revealed  that the order relation between 
quantum reals significantly depends on the underlying conditionals.  
In particular, we have completely characterized the operational meanings of 
those order relations in terms of joint probability obtained by the successive
projective measurements of arbitrary two observables.
Those characterizations clearly show their individual features 
and will play a fundamental role in future applications to quantum physics.


\end{document}